\documentclass[12pt]{article}

\usepackage{amssymb,amsmath,amsthm}
\newtheorem{Thm}{Theorem}

\newtheorem{Lem}{Lemma}
\newtheorem{Prop}{Proposition}
\theoremstyle{definition}

\newtheorem{Exam}{Example}

\newcommand{\bra}[1]{{\left\langle #1 \right|}}
\newcommand{\ket}[1]{{\left| #1 \right\rangle}}

\newcommand{\C}{\mbox{$\mathbb C$}}

\newcommand{\T}{\mbox{$\mathrm{tr}$}}

\title{On the geometric distance between quantum states with positive partial transposition and private states}
\author{
Jeong San Kim\footnote{Institute for Quantum Information Science, University of Calgary,
Alberta T2N 1N4, Canada. E-mail: jkim@qis.ucalgary.ca} \, and
Barry C. Sanders\footnote{Institute for Quantum Information Science, University of Calgary,
Alberta T2N 1N4, Canada.}
}
\date{\today}

\begin{document}

\maketitle

\begin{abstract}
 We prove an analytic positive lower bound for the geometric distance
between entangled positive partial transpose (PPT) states of a broad class
and any private state that delivers one secure key bit. Our proof holds for any Hilbert space
of finite dimension. Although our result is proven for a specific class of PPT states, we show that
our bound nonetheless holds for all known entangled PPT states with non-zero distillable key rates
whether or not they are in our special class.
\end{abstract}



\section{Introduction}

Whereas quantum entanglement is a non-local quantum correlation
among distinct systems providing various useful applications
such as quantum teleportation and dense coding~\cite{tele, dense}, it
is known that there are two different types of entanglement. One
of them is {\em free} entanglement, which can be distilled into
pure entanglement by means of {\em Local operations and Classical
Communications} (LOCC). Otherwise, the entanglement is said to be
{\em bounded}, if it is undistillable~\cite{PPTB}.

As most applications of quantum entanglement are based on the use of pure
entangled states, it is clear that having free entanglement
assures all of the tasks are possible. In quantum cryptography,
especially in quantum key distribution (QKD), secure key
distillation in QKD protocols has been considered to be closely
related with the distillation of pure entanglement~\cite{
SP}. Moreover, the existence of entanglement (whether it is free entanglement or not)
in a given state is known to be necessary for any protocol to distill secure
key from the state~\cite{CLL, CGLL}.

Although the class of {\em bound entangled
states} cannot be converted into a maximally entangled state
even in an asymptotic sense, bound entanglement
can be useful, in a catalytic way, in some quantum information processing (QIP)~\cite{HHH, M}.
However the use of bound entanglement as a resource in most QIP
still seems limited and strongly doubtful.

For any bipartite quantum systems, there is a simple necessary
condition for a given state to be separable. This is due to the
property that any separable state has Positive Partial
Transposition (PPT)~\cite{Peres}, and PPT is also known to be sufficient for separability
in $2\otimes 2$ or $2\otimes 3$ quantum systems~\cite{Horodeckis1}.
For entangled states, PPT is sufficient to guarantee bound entanglement:
Any PPT entangled state is a bound entangled state~\cite{PPTB},
whereas the existence of a bound entangled state with Negative Partial
Transposition (NPT) is still an open question.

In quantum cryptography, it has been recently shown that the class
of quantum states from which a secure key can be obtained just by
performing local measurement and classical communication is much
wider than that of maximally entangled states. Quantum states of
this class are called {\em private states}~\cite{HHHO1}, and
surprisingly, it has also been shown that there do exist some
classes of bound entangled states that can be asymptotically
approximated to private states~\cite{HHHO1, HHHO2, CCKKL, HA}.
In other words, even though the distillation
of pure entanglement from bound entangled states is not feasible,
some bound entangled states can be used as the resource in QKD
at least asymptotically. However, most examples of bound entangled states with non-zero
distillable key rates~\cite{HHHO1, HHHO2, CCKKL, HA} generally require
very large dimensional Hilbert space, in fact infinite dimensional
quantum system, to be approximated to private states.

Here, we address
the question of bound entanglement as being a resource for finite dimensional quantum systems
by showing that a broad class of bound entangled states are never arbitrarily close to private states
for any finite dimensional case.
By using trace norm value as the geometric
distance between quantum states, we provide an analytic positive
lower bound of the distance between a set of PPT states and
the set of private states in terms of the
dimension of the quantum system. We conjecture that this lower-bounded separation
of private states from bounded entangled states holds for all bounded entangled states, and
verify this conjecture for all known cases of bound entangled states with positive distillable
secure key rate~\cite{HHHO1, HHHO2, CCKKL, HA, HPHH}.

This paper is organized as follows. In Sec.~\ref{Sec:Private
States}, we recall the definition of private states and some
related propositions. In Sec.~\ref{Sec:Lower
Bound}, we propose an analytic lower bound of the trace norm distance
for a class of PPT states from the set of private states. In
Sec.~\ref{Sec:Examples}, we show that the proposed bound holds for
every PPT bound entangled state with non-zero distillable key rates
given in~\cite{HHHO1, HHHO2, CCKKL, HA, HPHH}, and we finally
summarize our results in Sec.~\ref{Sec:Conclusion}.

\section{Private States and Distillable Key Rate}\label{Sec:Private States}

A private state (or {\em pbit})~\cite{HHHO1, HHHO2}
 $\gamma_{ABA'B'}$ in $\mathcal B (\C^{2}\otimes \C^{2}\otimes \C^{d}\otimes \C^{d})$
is defined as
\begin{equation}
\gamma_{ABA'B'} = \frac{1}{2} \sum_{k,l=0}^{1}
\ket{kk}_{AB}\bra{ll}\otimes U_{kk}\rho_{A'B'} U^{\dagger}_{ll},
\label{eq:pdit}
\end{equation}
where $U_{jj}$'s and $\rho_{A'B'}$ are arbitrary unitary matrices
and a quantum state in subsystems $A'B'$, and  $d$ is the dimension of the
subsystems $A'$ and $B'$.
(We may assume that $A$ and $B$ have same dimension, otherwise we can take the larger one)

The set of private states is known to be the most general class of quantum states
that contains perfectly secure key~\cite{HHHO1}. In other words,
by performing local measurements on subsystems
$A$ and $B$ of  $\gamma_{ABA'B'}$ in Eq.~(\ref{eq:pdit}), one can obtain
one-bit secure key between $A$ and $B$. Conversely, if there is a
quantum state from which one can obtain one-bit secure key just by
performing local measurements on some subsystems (possibly whole
system) of the state, it has to be in forms of a private
state.

Because private states define all possible quantum states containing one-bit perfectly secure key,
a quantum state having high fidelity with a private state
can also be expected to behave similarly as a private state does: It would deliver
a secure key, although the key itself might not be perfectly secure.
This intuitive expectation is also shown to be true concerned with the relation between
fidelity of quantum states and their geometric distance. A state $\rho$ has a non-zero
distillable key rate, $K_D(\rho)>0$, if it is close enough to a
private state by means of trace norm distance~\cite{HHHO2}.
Moreover, $\rho$ itself does not need to be close enough to a
private state. Instead, if sufficiently many copies of $\rho$ can
be transformed into a state $\rho'$ by LOCC, and $\rho'$ is close
enough to a private state, it can be easily shown that not only
$\rho'$ but $\rho$ as well have non-zero distillable key
rates~\cite{CCKKL}.

The trace norm distance between a state $\rho$
and private states was also shown to have the following analytic
characterizations~\cite{HHHO2}.

\begin{Prop}\label{Prop:equiv1}
If the state $\rho_{ABA'B'} \in \mathcal{B}(\C^2 \otimes \C^2 \otimes \C^d
\otimes \C^{d'})$ written in the form
\begin{eqnarray}
\rho_{ABA'B'} &=&\sum_{i,j,k,l=0}^{1}\ket{ij}_{AB}\bra{kl} \otimes A_{ijkl} \nonumber \\
&=& \left[%
\begin{array}{cccc}
  A_{0000} & A_{0001} & A_{0010} & A_{0011} \\
  A_{0100} & A_{0101} & A_{0110} & A_{0111} \\
  A_{1000} & A_{1001} & A_{1010} & A_{1011} \\
  A_{1100} & A_{1101} & A_{1110} & A_{1111}
\end{array}%
\right], \label{eq:Dbound}
\end{eqnarray}
fulfills
$\|\rho_{ABA'B'}-\gamma_{ABA'B'} \| \le \varepsilon$ for some pbit $\gamma_{ABA'B'}$ and
$0<\varepsilon<1$, then there holds $\|A_{0011}\|\ge
1/2-\varepsilon$. Here, $\|\cdot\|$ is the trace norm defined as
$\|A\|=\T \sqrt{A^{\dagger}A}$, for any operator $A$.
\end{Prop}

\begin{Prop}\label{Prop:equiv2}
If the state $\rho_{ABA'B'} \in \mathcal{B}(\C^2 \otimes \C^2 \otimes \C^d
\otimes \C^{d'})$ written in the form
$\rho_{ABA'B'}=\sum_{i,j,k,l}\ket{ij}_{AB}\bra{kl} \otimes A_{ijkl}$ fulfills
$\|A_{0011}\|\ge 1/2-\varepsilon$ for $0<\varepsilon<1$, then
there exists a pbit $\gamma_{ABA'B'}$ such that $\|\rho_{ABA'B'}-\gamma_{ABA'B'}\| \le
\delta(\varepsilon)$ with $\delta(\varepsilon)$ vanishing, when
$\varepsilon$ approaches zero.
\end{Prop}

With respect to the trace norm value of a certain block $A_{0011}$ as well as
its difference from $1/2$ in the propositions above,
let us consider a purification $\ket{\psi}_{ABA'B'E}$ of $\rho_{ABA'B'}$ such that
\begin{equation}
\T_{E}\left( \ket{\psi}_{ABA'B'E}\bra{\psi}\right) =\rho_{ABA'B'}.
\end{equation}
By straightforward calculation~\cite{HHHO2},
we have
\begin{equation}
\|A_{0011}\|=\sqrt{p_{00}p_{11}} {\mathcal {F}}\left( \rho^E_{00}, \rho^E_{11}\right).
\label{A0011}
\end{equation}
where $p_{ii}=\T\left[\left(\ket{ii}_{AB}\bra{ii} \otimes I_{A'B'E} \right) \ket{\psi}_{ABA'B'E}\bra{\psi}\right]$
is the probability of the outcome state $\ket{ii}_{AB}$ for $i=0,1$ by the local measurement of $A$ and $B$,
$\rho^E_{ii}$ is the resulting state on the subsystem $E$ corresponding to the outcome $\ket{ii}_{AB}$ on $AB$,
that is,
\begin{equation}
\rho^E_{ii}=\T_{ABA'B'}\left[\left(\ket{ii}_{AB}\bra{ii} \otimes I_{A'B'E} \right) \ket{\psi}_{ABA'B'E}\bra{\psi}\right]/p_{ii},
\end{equation}
and ${\mathcal {F}}\left( \rho^E_{00}, \rho^E_{11}\right)$ is the fidelity of $\rho^E_{00}$ and $\rho^E_{11}$,
defined as
\begin{equation}
{\mathcal {F}}\left( \rho^E_{00}, \rho^E_{11}\right)=\T \sqrt{\sqrt{\rho^E_{00}}\rho^E_{11}\sqrt{\rho^E_{00}}}.
\end{equation}
By assuming the worst case scenario, that is, the eavesdropper Eve has the purification of $\rho_{ABA'B'}$,
$\|A_{0011}\|=1/2$ happens if and only if $p_{00}=p_{11}=1/2$, and, at the same time,
${\mathcal {F}}\left( \rho^E_{00}, \rho^E_{11}\right)=1$.
In other words, the only possible outcomes are $\ket{00}_{AB}$ and $\ket{11}_{AB}$ with the same probability, which provide perfect
correlation between $A$ and $B$, and this correlation is independent of Eve because $\rho^E_{00}$ and $\rho^E_{11}$ are identical.
This implies that $\rho_{ABA'B'}$ contains one-bit perfectly secure key, and thus it is a private state~\cite{HHHO1}.

Moreover, Proposition~\ref{Prop:equiv1} together with
Proposition~\ref{Prop:equiv2} give us a quantitative relation
between the distance of a quantum state $\rho_{ABA'B'}$ from a
private state and the trace norm value $\|A_{0011}\|$.
Thus, the lower bound of the trace norm distance
between a quantum state and private states can be now rephrased as an
upper bound of $\|A_{0011}\|$.

\section{PPT states and the lower bound of the distance}\label{Sec:Lower Bound}

Before we provide an analytic lower bound, let us consider the
negative eigenvalues of bipartite pure states that might arise
after partial transposition. For any pure state $\ket{\psi}_{AB} \in
\C^{d}\otimes \C^{d}$ with its Schmidt decomposition
\begin{equation}
\ket{\psi}_{AB}=\sum_{i=0}^{d-1}a_{i}\ket{ii}_{AB},~a_{i} \geq
0,~\sum_{i=0}^{d-1}a_{i}^{2}=1, \label{eq:pure state}
\end{equation}
the negative eigenvalues of the partially
transposed state ${\left(\ket{\psi}_{AB}\bra{\psi}
\right)}^{\Gamma}$ can be $- a_{i}a_{j}$ for $i \neq
j$~\cite{KDS}. Thus, the largest {\em negativity}~\cite{VW} (the sum of the absolute values of all negative eigenvalues) can be
achieved when $a_{i}=1/\sqrt{d}$ for all $i=0,\cdots,d-1$,
and it is $\sum_{i<j}1/d= (d-1)/2$. In this case, the sum
of positive eigenvalues is $(d+1)/2$, since
$\T{({\left(\ket{\psi}_{AB}\bra{\psi}
\right)}^{\Gamma})}=\T{(\ket{\psi}_{AB}\bra{\psi} )}=1$.
Thus, we have
\begin{equation}
\T{\left({\left(\ket{\psi}_{AB}\bra{\psi} \right)}^{\Gamma}\right)}=
\emph{P} - \emph{N},
\end{equation}
where $\emph{P}$ and $\emph{N}$ are the sum
of the absolute values of positive and negative eigenvalues of
${\left(\ket{\psi}_{AB}\bra{\psi} \right)}^{\Gamma}$ respectively,
and
\begin{equation}
\emph{P}\leq\frac{d+1}{2},~~ \emph{N}\leq\frac{d-1}{2}.
\label{eq:negativity sum}
\end{equation}

Now, let $\ket{\phi^{\pm}}$ and $\ket{\psi^{\pm}}$ be the Bell
states in $\C^2\otimes \C^2$; then we have the following lemma.

\begin{Lem}\label{Lem:Ubound}
For a state $\rho_{ABA'B'} \in \mathcal{B}(\C^2 \otimes \C^2
\otimes \C^d \otimes \C^{d})$,
\begin{align}
\rho_{ABA'B'} =&\ket{\phi^+}_{AB}\bra{\phi^+} \otimes \sigma_0
+\ket{\phi^-}_{AB}\bra{\phi^-} \otimes \sigma_1\nonumber\\
&+(\ket{\psi^+}_{AB}\bra{\psi^+} + \ket{\psi^-}_{AB}\bra{\psi^-})
\otimes \sigma_2, \label{belldia}
\end{align}
with arbitrary states $\sigma_i$, $i=0, \ldots, 2$ on subsystem $A'B'$,
if $\rho$ has PPT, then $\|\sigma_0 - \sigma_1\|$  has an upper
bound depending on the dimension $d$,
\begin{equation}
\|\sigma_0 - \sigma_1\| \le 1-\frac{1}{d+1}.
\label{eq:Ubound}
\end{equation}
\end{Lem}

Here, we note that any bipartite state can be transformed into
the class of quantum states in Lemma~\ref{Lem:Ubound} by applying
local depolarization and mixing operations~\cite{DCLB} on its two-qubit subsystems $AB$.
Thus, any PPT state can also be transformed into this class with PPT property
preserved. Furthermore, $\rho_{ABA'B'}$ in Eq.~(\ref{belldia}) has
the matrix form,
\begin{equation}
\rho_{ABA'B'} = \frac{1}{2}\left[
\begin{array}{cccc}
  \sigma_{0}+\sigma_{1}& 0 & 0 & \sigma_{0}-\sigma_{1}\\
  0 & 2\sigma_{2} & 0&0 \\
  0 & 0 & 2\sigma_{2}&0 \\
  \sigma_{0}-\sigma_{1}& 0 & 0 & \sigma_{0}+\sigma_{1}\\
\end{array}
\right], \label{eq:rho_matrix}
\end{equation}
where the trace norm value of its upper-right block $A_{0011}$ is
\begin{equation}
\|A_{0011}\|=\frac{\|\sigma_0 - \sigma_1\|}{2}.
\label{tracenorm}
\end{equation}
If $\rho_{ABA'B'}$ has PPT, then Lemma~\ref{Lem:Ubound} implies that
Eq.~(\ref{tracenorm}) is bounded above by $1/2-1/\left[2(d+1)\right]$.
Thus, by Proposition~\ref{Prop:equiv1}, the
trace norm distance between $\rho_{ABA'B'}$ and any private state
is bounded below by $1/\left[2(d+1)\right]$.

\begin{proof}[Proof of Lemma~\ref{Lem:Ubound}]
$(\sigma_0 -\sigma_1)^{\Gamma}$ is hermitian so it can have a diagonal representation as
\begin{equation}
 (\sigma_0 -
\sigma_1)^{\Gamma}=\sum_{j}\lambda_{j}\ket{x_j}\bra{x_j},
\end{equation}
where $\lambda_j$ are the eigenvalues (not necessarily non-negative)
with corresponding eigenvectors $\ket{x_j}$, then we have
\begin{equation}
\sigma_0 - \sigma_1
=\sum_{j}\lambda_{j}(\ket{x_j}\bra{x_j})^{\Gamma}.
\end{equation}
As $\rho$ has PPT, $\rho^{\Gamma}\geq 0$, which is
equivalent to $2\sigma_2^{\Gamma}\pm (\sigma_0 -
\sigma_1)^{\Gamma} \geq 0$~\cite{CCKKL}, we have
\begin{equation}
\sum_{j}|\lambda_{j}|=\|(\sigma_0-\sigma_1)^{\Gamma}\|\leq
2\|\sigma_{2}^{\Gamma}\|=2\|\sigma_{2}\|.
\end{equation}
Now, suppose $\|\sigma_0 - \sigma_1\| > 1-\frac{1}{d+1}$, then we have
\begin{align}
1-\frac{1}{d+1}<
\|\sigma_0 +
\sigma_1\|=1-2\|\sigma_{2}\|, \label{eq:Ubound01}
\end{align}
and thus, $\sum_{j}|\lambda_{j}|\leq 2\|\sigma_{2}\|<
\frac{1}{d+1}$. Hence
\begin{align}
\|\sigma_0 - \sigma_1\|
=&\|\sum_{j}\lambda_{j}(\ket{x_j}\bra{x_j})^{\Gamma}\|\nonumber\\
\leq &\sum_{j}|\lambda_{j}| \left\|(\ket{x_j}\bra{x_j})^{\Gamma}\right\|\nonumber\\
=& \sum_{i}|\lambda_{i}|(|\emph{P}|+|\emph{N}|)\nonumber\\
\leq& \sum_{j}|\lambda_{j}|d\nonumber\\
<& \|\sigma_0 - \sigma_1\|,
\end{align}
where the second inequality is due to Eq.~(\ref{eq:negativity
sum}). However, this is a contradiction, and thus,
\begin{equation}
\|\sigma_0 - \sigma_1\| \le 1-\frac{1}{d+1}.
\end{equation}
\end{proof}

Even though partial transposition (as well as full transposition)
of an operator strongly depends on the basis of its matrix
representation, the eigenvalues of the partially transposed
operator are independent from the choice of basis. Furthermore,
the proof of Lemma~\ref{Lem:Ubound} is only concerned with the
eigenvalues of the partially transposed quantum state. For this reason, the result of
Lemma~\ref{Lem:Ubound} is true not only for the private states
with a certain basis, but for any private state regardless
of the basis choice as well.

Now, we propose a more general class of PPT states, and prove that
the lower bound of the trace norm distance obtained in
Lemma~\ref{Lem:Ubound} is still valid for this class.
\begin{Thm}\label{Thm:Dbound}
For any state $\rho_{ABA'B'} \in \mathcal{B}(\C^2 \otimes \C^2 \otimes \C^d
\otimes \C^{d})$ with
\begin{eqnarray}
\rho_{ABA'B'} &=&\sum_{i,j,k,l=0}^{1}\ket{i,j}\bra{k,l} \otimes A_{ijkl} \nonumber \\
&=& \left[%
\begin{array}{cccc}
  A_{0000} & A_{0001} & A_{0010} & A_{0011} \\
  A_{0100} & A_{0101} & A_{0110} & A_{0111} \\
  A_{1000} & A_{1001} & A_{1010} & A_{1011} \\
  A_{1100} & A_{1101} & A_{1110} & A_{1111}
\end{array}%
\right], \label{eq:Dbound}
\end{eqnarray}
if $\rho_{ABA'B'}$ has PPT and its upper-right block $A_{0011}$ is
hermitian, then there exists a positive lower bound of the trace
norm distance between $\rho_{ABA'B'}$ and any private state $\gamma_{ABA'B'}$,
\begin{equation}
\|\rho_{ABA'B'}-\gamma_{ABA'B'} \| \geq  \frac{1}{2(d+1)}.
\label{eq.lbound}
\end{equation}
\end{Thm}

\begin{proof}
As $A_{0011}$ is hermitian, we have
$A_{1100}=A_{0011}^\dagger=A_{0011}$. By applying
depolarization and mixing operations, $\rho$ can be transformed to
\begin{equation}
\tilde{\rho} =\frac{1}{2}\left[%
\begin{array}{cccc}
  A_{0000}+A_{1111}& 0 & 0 & 2A_{0011} \\
  0 & A_{0101}+A_{1010} & 0&0 \\
  0 & 0 & A_{0101}+ A_{1010}&0 \\
  2A_{0011}& 0 & 0 & A_{0000}+A_{1111}
\end{array}%
\right]. \label{eq:rho_d}
\end{equation}

Because local depolarization and mixing operations are LOCC,
$\tilde{\rho}$ still has PPT. Furthermore, now $\tilde{\rho}$ is
of the form in Eq.~(\ref{eq:rho_matrix}); therefore $\|A_{0011}\|$
has an upper bound $\frac{1}{2}-\frac{1}{2(d+1)}$ by
Lemma~\ref{Lem:Ubound}.
Thus, by Proposition \ref{Prop:equiv1}, there is a lower
bound $\frac{1}{2(d+1)}$ of the trace norm distance between $\rho$
and any private state.
\end{proof}

In fact, it can be easily seen that the hermitian condition of
$A_{0011}$ in Theorem~\ref{Thm:Dbound} is equivalent to that
$\rho$ can be transformed into $\rho'$, a Bell-diagonal block
matrix, with $A_{0011}$ being untouched. Thus,
Theorem~\ref{Thm:Dbound} deals with the most general class of
quantum states where Lemma~\ref{Lem:Ubound} can be directly
applied.

\section{PPT bound entangled states with non-zero distillable key rates}\label{Sec:Examples}
In this section, we consider all known examples of PPT bound entangled states
of bipartite quantum systems with non-zero distillable key
rates~\cite{HHHO1, HHHO2, CCKKL, HA, HPHH}, and provide a positive lower bound of
the trace norm distance from private states.

\begin{Exam}(Horodecki et al~\cite{HHHO1, HHHO2})

We first consider the PPT states with $K_D>0$ presented in~\cite{HHHO1,HHHO2}.
Let
\begin{equation}
\rho=\dfrac{1}{N}
\left[%
\begin{array}{cccc}
   \left[p({\tau_1+\tau_0})\right]^{\otimes m} &0&0& \left[p({\tau_1-\tau_0})\right]^{\otimes m} \\
0& \left[(1-2p)\tau_0\right]^{\otimes m}&0&0 \\
0&0&\left[(1-2p)\tau_0\right]^{\otimes m}& 0\\
\left[p{({\tau_1-\tau_0})}\right]^{\otimes m} &0&0& \left[p({\tau_1+\tau_0})\right]^{\otimes m}  \\
\end{array}%
\right],
\label{eq:ex1}
\end{equation}
where $N=2(2p)^m+2(1-2p)^m$, $\tau_0=\varrho_s^{\otimes l}$,
$\tau_1=[(\varrho_a+\varrho_s)/2]^{\otimes l}$,
$\varrho_s=2P_\text{sym}/({d^2+d})$ and
$\varrho_a=2P_\text{as}/({d^2-d})$ with
the antisymmetric projector $P_\text{as}$ and symmetric projector $P_\text{sym}$ in $\C^d \otimes \C^d$ system.
For sufficiently large $l$, $m$ and $d$, $\rho$ is known to have a non-zero
distillable key rate $K_D(\rho)>0$. At the same time, $\rho$
can be also shown to have PPT  with a choice of $p\in [0,1/3]$,
thus $\rho$ is a PPT bound entangled state with $K_D(\rho)>0$.

It can be directly checked that the upper-right block of $\rho$ in Eq.~(\ref{eq:ex1})
\begin{equation}
A_{0011}=\frac{1}{N}\left[p({\tau_1-\tau_0})\right]^{\otimes m}
\end{equation}
is hermitian since both $\tau_0$ and $\tau_1$ are hermitian operators.
Thus, $\rho$ satisfies the condition of Theorem~\ref{Thm:Dbound}: There is a positive lower bound for 
the trace norm distance between
$\rho_{ABA'B'}$ and any private state in finite dimensional quantum system, that is, for any finite
$l$, $m$ and $d$ in Eq.~(\ref{eq:ex1}).
\end{Exam}

\begin{Exam}(Chi et al.~\cite{CCKKL})

There were two different classes of PPT bound entangled states with non-zero distillable key
rates proposed in~\cite{CCKKL}.

First, let us consider a quantum state $\rho_{ABA'B'} \in \mathcal B\left( \C^2 \otimes\C^2 \otimes\C^2 \otimes\C^2\right)$ such that
\begin{equation}
\rho_{ABA'B'}=\frac{1}{2}\left[%
\begin{array}{cccc}
  {\sigma_0+\sigma_1} & 0 & 0 &  {\sigma_0-\sigma_1} \\
  0 & 2\sigma_2 & 0 & 0 \\
  0 & 0 & 2\sigma_2 & 0 \\
  {\sigma_0-\sigma_1} & 0 & 0 &  {\sigma_0+\sigma_1} \\
\end{array}%
\right],
\label{eq:ex3}
\end{equation}
where
\begin{align}
\sigma_0=&p\left(\ket{\phi^+}\bra{\phi^+}+\ket{01}\bra{01}\right),~
\sigma_1=p\left(\ket{\phi^-}\bra{\phi^-}+\ket{10}\bra{10}\right), \nonumber\\
\sigma_2=&\frac{p}{\sqrt{2}}\left(\ket{01}\bra{01}+\ket{10}\bra{10}\right) 
+q\ket{00}\bra{00}+r\ket{11}\bra{11},
\end{align}
and $\ket{\phi^{\pm}}=\frac{1}{\sqrt{2}}\left(\ket{00}\pm\ket{11}\right)$.

For $p=[1-2(q+r)]/(4+2\sqrt{2})$, $q>0$, $r>0$, and $0\leq q+r<({2-\sqrt{2}})/{8}$,
$\rho$ is has PPT . Furthermore, by using
an entanglement distillation protocol~\cite{BDSW} on the subsystem $AB$,
$n$ copies of $\rho_{ABA'B'}$
can be transformed to a quantum state
$\rho'_{ABA'B'} \in \mathcal{B}(\C^2 \otimes \C^2 \otimes \C^d \otimes \C^{d})$ with $d=2^n$
such that
\begin{equation}
\rho'_{ABA'B'}=\frac{1}{N}\left[%
\begin{array}{cccc}
  \left({\sigma_0+\sigma_1}\right)^{\otimes n} & 0 & 0 &  \left({\sigma_0-\sigma_1}\right)^{\otimes n} \\
  0 & \left(2\sigma_2\right)^{\otimes n} & 0 & 0 \\
  0 & 0 & \left( 2\sigma_2\right)^{\otimes n} & 0 \\
  \left({\sigma_0-\sigma_1}\right)^{\otimes n} & 0 & 0 &  \left({\sigma_0+\sigma_1}\right)^{\otimes n} \\
\end{array}%
\right],
\label{eq:ex3'}
\end{equation}
where $N=2^{n+1}\left[\left(2p\right)^n+\left(\sqrt{2}p+2q+2r\right)^n \right]$ is the normalization factor.

$\rho'_{ABA'B'}$ was shown to be approximated to a private state as $n$ is getting larger, so that
it has a non-zero distillable key rate.
Furthermore, $\rho'_{ABA'B'}$ still has PPT since the entanglement distillation protocol is LOCC,
and $\rho'_{ABA'B'}$ is transformed from many copies of PPT state $\rho_{ABA'B'}$.

However, the upper-right block of $\rho'_{ABA'B'}$ is
\begin{equation}
A_{0011}=\frac{1}{N}\left({\sigma_0-\sigma_1}\right)^{\otimes n},
\end{equation}
which is a hermitian operator. Thus, $\rho'_{ABA'B'}$ satisfies the condition of Theorem~\ref{Thm:Dbound},
and this class of PPT states with non-zero distillable key rates
has a lower bound of the trace norm distance for any finite number of $n$.

Chi et al.'s other class of PPT states with non-zero distillable key rates was constructed as
\begin{align}
\rho_{ABA'B'}=&\ket{\phi^+}\bra{\phi^+}\otimes \sigma_0 +\ket{\phi^-}\bra{\phi^-} \otimes \sigma_1
\nonumber\\
&+\ket{\psi^+}\bra{\psi^+} \otimes \sigma_2+\ket{\psi^-}\bra{\psi^-} \otimes \sigma_3,
\label{eq:ex4}
\end{align}
where
\begin{align}
\sigma_0=&p\left(\ket{\phi^+}\bra{\phi^+}+\ket{01}\bra{01}\right),~ 
\sigma_1=p\left(\ket{\phi^-}\bra{\phi^-}+\ket{10}\bra{10}\right), \nonumber\\
\sigma_2=&\sqrt{2}p\ket{x_0}\bra{x_0}+q\ket{00}\bra{00},~ 
\sigma_3=\sqrt{2}p\ket{x_1}\bra{x_1}+q\ket{00}\bra{00},
\label{eq:ex44}
\end{align}
and some orthonormal states $\ket{x_0}$ and $\ket{x_1}$ such that
\begin{equation}
\ket{x_0}\bra{x_0}+\ket{x_1}\bra{x_1}=\ket{01}\bra{01}+\ket{10}\bra{10}.
\end{equation}
For $p=(1-2q)/(4+2\sqrt{2})$ and $0<q<{(2-\sqrt{2})}/{8}$, it was shown that $\rho_{ABA'B'}$
has PPT, and it can be approximated to a private state by using the entanglement distillation protocol~\cite{BDSW}.
Similarly with the class of  Eq.~(\ref{eq:ex3}), it can be checked that the approximation in any
finite dimensional quantum system still has a lower bound of the trace norm distance.
\end{Exam}

Later, there was another class of PPT bound entangled states with non-zero distillable key rates was proposed~\cite{HA}
based on the similar construction with~\cite{HHHO1, HHHO2}. Ii can also be directly checked that
the states in this class satisfy the condition of Theorem~\ref{Thm:Dbound}, and thus, they have a positive
lower bound of the trace norm distance from any private state.

Now, let us consider the examples of PPT bound entangled states proposed in~\cite{HPHH}.
\begin{Exam}(Horodecki et al.~\cite{HPHH})
For any quantum state $\rho_{ABA'B'}$ in $\mathcal{B}(\C^2 \otimes \C^2 \otimes \C^d \otimes \C^{d})$
such that
\begin{equation}
\rho_{ABA'B'}=\frac{1}{2}\left[%
\begin{array}{cccc}
  p\sqrt{X_1X_1^{\dagger}} & 0 & 0 & pX_1 \\
  0 & (1-p)\sqrt{X_2X_2^{\dagger}} & (1-p)X_2 & 0\\
  0 & (1-p)X_2^{\dagger} & (1-p)\sqrt{X_2^{\dagger}X_2} & 0 \\
  pX_1^{\dagger} & 0 & 0 & p\sqrt{X_1^{\dagger}X_1} \\
\end{array}%
\right],
\label{eq:low_dim1}
\end{equation}
where $X_1$ and $X_2$ are arbitrary operator with trace norm one, and $0\leq p \leq 1$,
it is shown that the distillable key rate $K_D (\rho_{ABA'B'})$ fulfills $K_D (\rho_{ABA'B'})\geq1-h(p)$
where $h(p)$ is the binary entropy of distribution $\{p, 1-p \}$~\cite{HPHH}.
To construct a PPT state of the form Eq.~(\ref{eq:low_dim1}), let
\begin{equation}
X_{1}=\frac{1}{\|W_{U} \| } W_{U},~ X_{2}=\frac{W_{U}^{\Gamma}}{\|W_{U}^{\Gamma}\|},
\end{equation}
where
\begin{equation}
W_{U}=\sum_{i,j}u_{ij}\ket{ij}\bra{ji}, \label{eq:W_U}
\end{equation}
$u_{ij}$ are the elements of a unitary matrix $U$ on $\C^{d}$, and
 $W_{U}^{\Gamma}$ is the partially transposed matrix of $W_U$.
We also let the probabilities be
\begin{equation}
p=\frac{\|W_{U}\|}{\|W_{U}\|+\|W_{U}^{\Gamma}\|},~ 1-p= \frac{\|W_{U}^{\Gamma}\|}{\|W_{U}\|+\|W_{U}^{\Gamma}\|};
\label{eq:probabilities}
\end{equation}
then, by direct observation, $\rho_{ABA'B'}$ is invariant under partial transposition and obviously PPT.
Furthermore, we have
\begin{equation}
\frac{p}{1-p}=\frac{\|W_{U}\|}{\|W_{U}^{\Gamma}\|}=\frac{\sum_{i,j}|u_{ij}|}{d}.
\label{eq:security}
\end{equation}
Thus, a non-zero distillable key rate of $\rho_{ABA'B'}$ that is equivalent
to $p \neq \frac{1}{2}$ can be obtained from any choice of the
unitary matrix $U$ satisfying $\sum_{i,j}|u_{ij}|>d$.

Now let us consider the trace norm distance of $\rho_{ABA'B'}$ from private states. Unfortunately, $\rho_{ABA'B'}$
does not satisfy the condition of Theorem~\ref{Thm:Dbound}, since its upper-right block $A_{0011}=\frac{1}{2}pX_1$
is not hermitian. However, in this case, we can directly
evaluate an upper bound of $\|A_{0011}\|$. By noticing that $\|X_1
\|=1$ and $\|W_{U}^{\Gamma} \|=d$, we have
\begin{align}
\|A_{0011}\|=\|\frac{1}{2}pX_1\|
=& \frac{1}{2}- \frac{\|W_{U}^{\Gamma}\|}{2(\|W_{U}\|+\|W_{U}^{\Gamma}\|)}\nonumber\\
=& \frac{1}{2}- \frac{d}{2(\|W_{U}\|+d)},
\end{align}
and thus, $\|A_{0011}\|$ attains its maximum value when $\|W_{U}\|$ is the largest.
We also have
\begin{equation}
\|W_{U}\|=\T\sqrt{W_{U}{W_U}^{\dagger}}=\sum_{i,j}|u_{ij}|;
\end{equation}
therefore, by simple calculus, we have $\|W_{U}\|\leq d\sqrt{d}$, and
\begin{align}
\|A_{0011}\| \leq & \frac{1}{2}- \frac{1}{2(\sqrt{d}+1)}\nonumber\\
\leq& \frac{1}{2}- \frac{1}{2(d+1)},
\label{eq:ex3ubound}
\end{align}
where the second inequality indicates the upper bound obtained in Lemma~\ref{Lem:Ubound}.
Thus, the trace norm distance between $\rho_{ABA'B'}$ and
private states is still satisfied by the lower bound obtained in Theorem~\ref{Thm:Dbound},
\begin{align}
\|\rho_{ABA'B'}-\gamma_{ABA'B'} \| \geq&  \frac{1}{2(\sqrt{d}+1)}\nonumber\\
\geq& \frac{1}{2(d+1)},
\end{align}
for any private state $\gamma_{ABA'B'}$ in
$\mathcal{B}(\C^2 \otimes \C^2 \otimes \C^d \otimes \C^{d})$.
\end{Exam}
In this section, we have considered all known examples of PPT
bound entangled states in bipartite quantum systems with non-zero
distillable key rates. We have remarked that most of the examples
belong to the class proposed in Theorem~\ref{Thm:Dbound}, and the
only exceptional case~\cite{HPHH} is also shown to have the
same bound. (In fact, a coarser bound). Here, we conjecture that
the lower bound of the trace norm distance proposed in
Theorem~\ref{Thm:Dbound} is true for any PPT state.


\section{Conclusions}\label{Sec:Conclusion}

We have provided an analytic positive lower bound of the distance for
a class of PPT states from private states in terms of the dimension of quantum systems, and
we have further shown that this lower bound holds for all known examples
of PPT bound entangled states with non-zero distillable key rates.

Our result implies that most PPT bound entangled states, including
all the examples of non-zero distillable key rates, are
geometrically separated from private states in any finite
dimensional quantum system. This is a strong clue for our conjecture
that every PPT state is distantly separated from private states
in any finite dimensional Hilbert space. The conjecture directly leads to the answer
for an important question: Whether or not, the operational
meaning of bound entanglement in QKD is only asymptotic.
Thus, our result and conjecture will play a central role to clarify the
essential difference between free and bound entangled states in quantum cryptography in both operational
and analytical ways.
Besides quantum cryptography, our result will also provide a rich reference with useful tools
in studying the geometric structure of bounded operators with certain properties in Hilbert space.

\section*{Acknowledgments}
This work is supported by iCORE, CIFAR Associateship, MITACS, and US army.


\end{document}